\documentclass[10pt,twocolumn,a4paper]{SICE14}

\usepackage{color}
\usepackage{amsmath}
\usepackage{amsfonts,amssymb}

\newtheorem{problem}{\noindent\hspace{1em}\bf Problem}

\newcommand{\ssr}[4]{\left[ 
    \begin{array}{c|c}
      #1&#2\\ \hline 
      #3&#4
    \end{array}
  \right]}

\title{Sampled-Data $H^{\infty}$ Design of Coupling Wave Cancelers in\\
Single-Frequency Full-Duplex Relay Stations
}

\author{Masaaki Nagahara${}^{1\dagger \ddagger}$, Hampei Sasahara${}^{2\ddagger}$, Kazunori Hayashi${}^{3\ddagger}$, and Yutaka Yamamoto${}^{4\ddagger}$}
\speaker{Hampei Sasahara}

\affils{${}^{\ddagger}$Graduate School of Informatics, Kyoto University, Kyoto, Japan\\
${}^{1}$nagahara@ieee.org,
${}^{2}$sasahara.h@acs.i.kyoto-u.ac.jp,\\
${}^{3}$kazunori@i.kyoto-u.ac.jp,
${}^{4}$yy@i.kyoto-u.ac.jp\\
}
\abstract{
In this article, we propose sampled-data $H^\infty$ design of digital filters that cancel the
continuous-time effect of coupling waves in a single-frequency full-duplex relay station.
In this study, we model a relay station as a continuous-time
system while conventional researches treat it as a discrete-time system. 
For a continuous-time model, we
propose digital feedforward and feedback cancelers based on
the sampled-data control theory to cancel coupling waves taking intersample behavior into account.
Simulation results are shown to illustrate the effectiveness of the proposed method.
}


\begin{document}

\maketitle


\section{Introduction}
\label{sec:intro}

In wireless communications, 
{\em relay stations} are used to relay radio signals
between radio stations that cannot directly communicate
with each other due to the signal attenuation.
On the other hand,
it is important to efficiently utilize the scarce bandwidth
due to the limitation of frequency resources
\cite[Chap.~1]{Gol}.
For this purpose,
{\em single-frequency network} is often preferable
in which signals with the same carrier frequency are transmitted through communication networks.
Then, a problem of {\em self-interference} caused by coupling waves arises 
in a full-duplex relay station in a single-frequency network \cite{Jai+11}.

Fig.~\ref{coupling} illustrates self-interference by coupling waves.
In this picture, radio signals with carrier frequency $f_1$ 
are transmitted from the base station (denoted by BS).
One terminal (denoted by T1) directly receives the signal from the base station,
but the other terminal (denoted by T2) 
is so far from the base station that they cannot communicate directly.
Therefore, a relay station (denoted by RS) is attached between them
to relay radio signals.
Then, radio signals with carrier frequency $f_1$
from the transmission antenna of the relay station
are fed back to the receiving antenna
directly or through reflection objects.
As a result, self-interference is caused in the relay station,
which may deteriorate the quality of communication
and, even worse, may destabilize the system.

For the problem of self-interference, 
adaptive methods have been proposed 
to cancel the effect of coupling waves:
a least mean square (LMS) adaptive filters
\cite{SakOkaHay06}, and adaptive array antennas 
\cite{NogHayKanSak12}.
In these studies, a relay station is modeled by
a discrete-time system, and the performance
is optimized in the discrete-time domain.
However, radio waves are in nature continuous-time
signals and hence the performance should be
discussed in the continuous-time domain.
In other words, one should take account of {\em intersample behavior}
for coupling wave cancelation.

In theory, if the signals are completely band-limited below the
Nyquist frequency,
then the intersample behavior can be restored 
from the sampled-data in principle \cite{Sha},
and the discrete-time domain approaches might work well.
However, the assumption of perfect band limitedness is
hardly satisfied in real signals;
real baseband signals are not fully band-limited
(otherwise they must be non-causal \cite[Chap.~1]{Skl}),
pulse-shaping filters, such as raised-cosine filters, do not act perfectly,
and the nonlinearity in electric circuits
adds frequency components beyond the Nyquist frequency.
One might think that if the sampling frequency is fast enough,
the assumption is almost satisfied and there is no problem.
But this is not true; firstly, the sampling frequency cannot be arbitrarily
increased in real systems, and secondly, even though the sampling
is quite fast, intersample oscillations may happen
in feedback systems \cite[Sect.~7]{Yam99}.

	\begin{figure}[t]
		\centering
		\scalebox{0.55}{\includegraphics{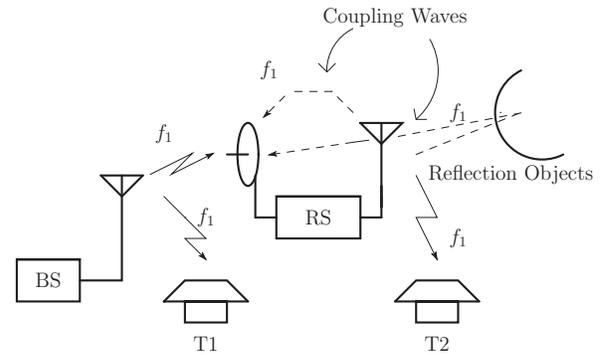}}
		\caption{Self-interference}
		\label{coupling}
	\end{figure}
To solve the problem mentioned above, 
we propose a new design method for coupling wave cancelation
based on the {\em sampled-data control theory}
\cite{CheFra,Yam99}.
We model the transmitted radio signals and coupling waves
as continuous-time signals, and optimize the 
worst case continuous-time error due to coupling waves
by a {\em digital} canceler.
This is formulated as a sampled-data $H^\infty$ optimal control problem,
which is solved via the fast-sampling fast-hold (FSFH) method
\cite{KelAnd92,YamMadAnd99}.
In this study, we consider two types of digital canceler:
feedforward and feedback cancelers.
For a feedforward canceler%
\footnote{A feedforward canceler was first reported in \cite{SasNagHayYam14}.},
we cancel self-interference by
a discrete-time (virtual) model of the coupling wave path
that is optimized via sampled-data $H^\infty$ optimization
\cite{YamNagKha12,NagYam13}.
For a feedback canceler,
we place a digital controller in the feedback loop
for stabilizing the feedback system as well as canceling
the self-interference.
This is formulated as a standard sampled-data $H^\infty$
control problem except for the time delay in the feedback loop,
which can be solved via FSFH as well.
Design examples are shown to illustrate the proposed methods.

The reminder of this article is organized as follows.
In Section \ref{sec:ff},
we formulate a design problem of feedforward cancelers.
In Section \ref{sec:fb}, we formulate a feedback canceler design problem
as a sampled-data $H^{\infty}$ optimal control problem,
which can be solved via FSFH approximation
described in Section \ref{sec:FSFH}.
In Section \ref{sec:sim}, simulation results are shown to illustrate the effectiveness of the proposed method.
In Section \ref{sec:conc}, we offer concluding remarks.

\subsection*{Notation}
Throughout this article, we use the following notation.
We denote by $L^2$ the Lebesgue space consisting of all square integrable real functions 
on $[0, \infty)$ endowed with $L^2$ norm $\|\cdot\|_{L^2}$,
and $\ell^2$ the space consisting of all square summable sequences,
with $\ell^2$ norm $\|\cdot\|_{\ell^2}$.
The symbol $t$ denotes the argument of time, $s$ the argument of Laplace transform and $z$ the argument of $Z$ transform.
These symbols are used to indicate whether a signal or a system is of continuous-time or discrete-time.
The operator $e^{-Ls}$ with nonnegative real number $L$ denotes continuous-time delay operator
with delay time $L$.
A continuous-time (or discrete-time) system $G$ with transfer function 
$C(sI-A)^{-1}B+D$ (or $C(zI-A)^{-1}B+D$) is denoted by
\[
G = \ssr{A}{B}{C}{D}.
\]

%
%
\section{Feedforward cancelers}
\label{sec:ff}
In this section, we formulate
the design problem of feedforward cancelers.

Fig.~\ref{outside} shows the block diagram of a relay station
using an amplify and forward relaying protocol
with a coupling wave path and a digital feedforward canceler.
	\begin{figure}[t]
		\scalebox{0.85}{\includegraphics{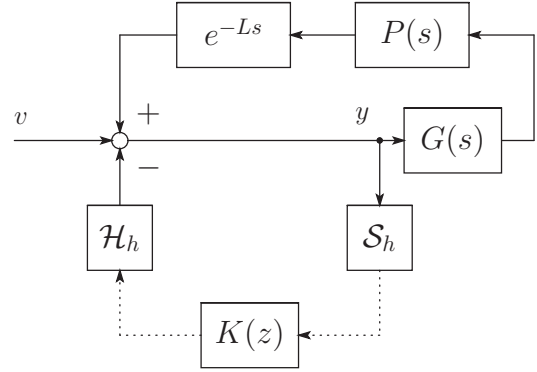}}
		\caption{Feedforward canceler}
		\label{outside}
	\end{figure}
In Fig.~\ref{outside}, 
continuous-time signals are represented in solid lines and
discrete-time signals in dotted lines.
We model the relay station by a continuous-time linear time-invariant (LTI) system
with transfer function $G(s)$.
The characteristic of the coupling wave path is also modeled by a
continuous-time LTI system denoted by $P(s)$ with time delay $e^{-Ls}$
where $L$ is a fixed nonnegative real number (i.e. $L$ is a fixed delay time).
Note that $P(s)$ will be time-varying in general due to the Doppler shift caused by the movement of reflection objects, however, the assumption of the LTI system is valid if the movement is slow or the signal components coming from the moving objects are not dominant.
The digital canceler consists of three operators:
ideal sampler ${\mathcal S}_h$ with sampling period $h>0$,
digital filter $K(z)$, and zero-order hold ${\mathcal H}_h$ with the same sampling period $h$.
The ideal sampler is
defined by
\[
\begin{split}
  {\mathcal S}_h: \{y(t)\} &\mapsto \{y_d[n]\}:
                 y_d[n] = y(nh),\\
                 n &= 0,1,2,\ldots,
\end{split}
\]
and the zero-order hold is defined by
\[
\begin{split}
 {\mathcal H}_h:  \{u_d[n]\} &\mapsto \{u(t)\}:
                 u(t) = u_d[n],\\
                  t &\in [nh,(n+1)h), \quad n=0,1,2,\ldots.
\end{split}
\]

From Fig.~\ref{outside}, we have
\begin{equation}
 y = v + (e^{-Ls}PG - {\mathcal H}_hK{\mathcal S}_h)y. \label{eqn1}
\end{equation}
To model the characteristic of the input signal $y$,
we introduce a subset $FL^2\subset L^2$ defined by
\[
 FL^2 := \{y=Fw: w \in L^2, ~\| w \|_{L^2}=1\},
\]
where $F$ is a continuous-time LTI system with 
real-rational, stable, and strictly proper transfer function $F(s)$.
The transfer function is a frequency domain weighting function
that gives the frequency characteristic of $y$ \footnote{
Although the assumption that $y\in FL^2$ is artificial,
we here consider a much wider class of signals
than the band-limited signal class assumed in
Shannon's theorem.
}.
Note that this signal model allows non band-limited signals
such as rectangular waves.
For any $y\in FL^2$, we try to {\em uniformly} minimize the error
\[
 \begin{split}
  e&:=(e^{-Ls}PG-{\mathcal H}_hK{\mathcal S}_h)y\\
  &=(e^{-Ls}PG-{\mathcal H}_hK{\mathcal S}_h)Fw
 \end{split}
\]

In other words, we minimize the $H^\infty$ norm
of the error system (see Fig.~\ref{error})
\[
 {\mathcal{E}}(K) := (e^{-Ls}PG-{\mathcal H}_hK{\mathcal S}_h)F,
\]
that is,
\[
 \inf_{K: {\rm stable}} \|{\mathcal{E}}(K)\|_\infty
 =\inf_{K: {\rm stable}}\sup_{\substack{w\in L^2\\ \|w\|_{L^2}=1}} {\|\mathcal{E}}(K)w\|_{L^2}.
\]
	\begin{figure}[t]
		\scalebox{0.9}{\includegraphics{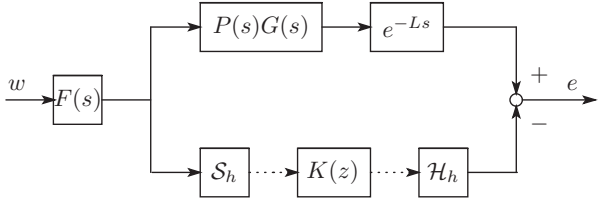}}
		\caption{Error system}
		\label{error}
	\end{figure}
The optimal discrete-time filter $K(z)$ can be obtained via
FSFH discretization \cite{KelAnd92,YamMadAnd99}.
For details of the design procedure, see  \cite{NagYam13}.

If we use the $H^\infty$ optimal filter $K(z)$ for the relay station 
that achieves sufficiently small $\|{\mathcal{E}}(K)\|_\infty$,
the effect of the coupling wave, $y-v$, is sufficiently
reduced.
In fact, we have the following theorem:
\begin{theorem}
	Assume $\|{\mathcal{E}}(K)\|_\infty\leq\gamma$ with $\gamma>0$.
	Then for any $y\in FL^2$ we have
	 $\|y-v\|_{L^2} \leq \gamma$.
\end{theorem}
\begin{proof}
	For any $y\in FL^2$, there exists $w\in L^2$ such that $y=Fw$ with $\|w\|_{L^2}=1$. 
	This and equation (\ref{eqn1}) give
	\[
	\begin{split}
	\|y-v\|_{L^2} &= \|(e^{-Ls}PG-{\mathcal H}_hK{\mathcal S}_h)y\|_{L^2} \\
			 &= \|(e^{-Ls}PG-{\mathcal H}_hK{\mathcal S}_h)Fw\|_{L^2} \\
			 &\leq \|{\mathcal{E}}(K)\|_\infty\\
			 &\leq \gamma.
	\end{split}
	\]
\end{proof}
This theorem motivates the proposed $H^\infty$ optimal design
for cancelation of coupling waves.

%
%
\section{Feedback cancelers}
\label{sec:fb}
The feedforward canceler shown in Fig.~\ref{outside}
works well when the gain of $G(s)$ is low.
If the gain of $G(s)$ is very high, the self-interference feedback loop
including the coupling wave path may become unstable.
Since the feedforward canceler design does not take the stability
into account, it cannot generally stabilize the feedback loop.
Therefore, we here consider a {\em feedback} canceler
to stabilize the feedback loop as well as reducing the effect of self-interference.
Fig.~\ref{model0} shows the block diagram of a relay station 
attached with a digital feedback canceler
${\mathcal{H}}_hK(z){\mathcal{S}}_h$.
The difference between this and the feedforward canceler
in Fig.~\ref{outside} is that the canceler is placed in the feedback loop.
	\begin{figure}[t]
		\scalebox{0.6}{\includegraphics{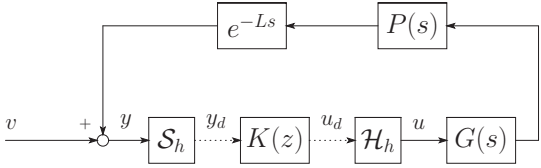}}
		\caption{Feedback canceler}
		\label{model0}
	\end{figure}

Our problem here is to
design the digital controller, $K(z)$, that stabilizes the feedback loop
and minimize the effect of self-interference, $z:=v-u$,
for any $v$.
We restrict the input continuous-time signal $v$ to the following
set:
\[
 WL^2 := \{v = Ww: w \in L^2, \|w\|_{L^2}=1\},
\]
where $W$ is a continuous-time LTI system with real-rational, stable,
and strictly proper transfer function $W(s)$.
The design problem is formulated as follows:
\begin{problem}
Design digital controller (canceler) $K(z)$ that stabilizes
the self-interference feedback loop and uniformly minimizes
the $L^2$ norm of the error $z=v-u$ for any $v\in WL^2$.
\end{problem}

This problem is reducible to a standard sampled-data $H^\infty$ control problem
\cite{CheFra,Yam99}.
To see this, let us consider the block diagram shown in Fig.~\ref{model}.
	\begin{figure}[t]
		\centering
		\scalebox{0.55}{\includegraphics{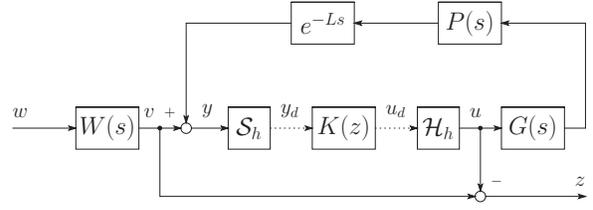}}
		\caption{Block diagram for feedback canceler design}
		\label{model}
	\end{figure}
Let $T_{zw}$ be the system from $w$ to $z$.
Then we have
\[
 z = v - u = T_{zw}w
\]
and hence uniformly minimizing $\|z\|_{L^2}$ for any $v\in WL^2$ is equivalent to
minimizing the $H^\infty$ norm of $T_{zw}$,
\begin{equation}
 \|T_{zw}\|_\infty = \sup_{w\in L^2, \|w\|_{L^2}=1} \|T_{zw}w\|_{L^2}.
 \label{eq:hinf_Tzw}
\end{equation}
Let $\Sigma(s)$ be a generalized plant given by
\[
 \Sigma(s) = \begin{bmatrix}W(s)&-1\\ W(s) & e^{-Ls}P(s)G(s)\end{bmatrix}.
\]
By using this, we have
\[
 T_{zw}(s) = {\mathcal{F}}(\Sigma(s), {\mathcal{H}}_hK(z){\mathcal{S}}_h),
\]
where ${\mathcal{F}}$ denotes the linear-fractional transformation (LFT)
\cite{CheFra}. Fig.~\ref{geplant} shows the block diagram of this LFT.
Then our problem is to find a digital controller $K(z)$ that minimizes
$\|T_{zw}\|_\infty$.
\begin{figure}[t]
	\centering
	\scalebox{0.7}{\includegraphics{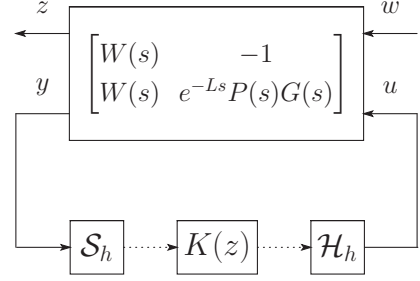}}
	\caption{LFT $T_{zw} = {\mathcal{F}}(\Sigma, {\mathcal{H}}_hK{\mathcal{S}}_h)$}
	\label{geplant}
\end{figure}
This is a standard sampled-data $H^\infty$ control problem,
and can be efficiently solved via
FSFH approximation.
We discuss this in the next section.

Note that if there exists a controller $K(z)$ that minimizes
$\|T_{zw}\|_\infty$, then the feedback system is stable
and the effect of self-interference $z=v-u$ is bounded by
the $H^\infty$ norm.
We summarize this as a theorem.
\begin{theorem}
Assume $\|T_{zw}\|_\infty \leq \gamma$ with $\gamma > 0$.
Then the feedback system shown in Fig.~\ref{model0} is stable,
and for any $v\in WL^2$ we have $\|v-u\|_{L^2}\leq \gamma$.
\end{theorem}
\begin{proof}
First, if the feedback system is unstable, then the $H^\infty$ norm becomes
unbounded.
Next, for $v\in WL^2$ there exists $w\in L^2$ such that 
$v=Ww$ and $\|w\|_{L^2}=1$.
Then, inequality $\|T_{zw}\|_\infty \leq \gamma$ gives
\[
 \|v-u\|_{L^2} = \|T_{zw}w\|_{L^2} \leq \|T_{zw}\|_\infty \|w\|_{L^2} \leq \gamma.
\] 
\end{proof}

%
%
\section{Fast-sample fast-hold approximation}
\label{sec:FSFH}
In this section, we review the method of FSFH approximation
for sampled-data $H^\infty$ optimal controller design.
The idea of FSFH is approximating a continuous-time $L^2$ signal
by a piecewise constant signal, which is generated by a fast hold
${\mathcal H}_{h/N}$ where $N$ is an integer greater than 2,
and evaluating the $L^2$ norm of an $L^2$ signal on the sampling points
generated by ${\mathcal S}_{h/N}$.
We call ${\mathcal S}_{h/N}$ and ${\mathcal H}_{h/N}$
a fast sampler and a fast hold, respectively.
That is, we connect the fast sampler and hold to the continuous-time
signals, $z$ and $w$, in Fig.~\ref{geplant} respectively, to make
a generalized plant with discrete-time input/output
as shown in Fig.~\ref{FSFHplant}.
	\begin{figure}[t]
		\centering
		\scalebox{0.65}{\includegraphics{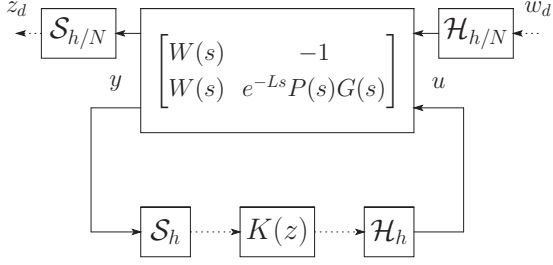}}
		\caption{FSFH approximation of $T_{zw}$}
		\label{FSFHplant}
	\end{figure}

Roughly speaking, if we take $N\rightarrow\infty$, then the frequency response
of the FSFH approximation ${\mathcal S}_{h/N}T_{zw}{\mathcal H}_{h/N}$
uniformly approaches that of sampled-data system $T_{zw}$,
see \cite{YamMadAnd99} for details.
Since the $H^\infty$ norm defined in \eqref{eq:hinf_Tzw} is equivalent to
the maximum gain of the frequency response gain of the sampled-data
system $T_{zw}$ (defined via lifting)~\cite[Chap.~13]{CheFra}, 
we can obtain an approximated solution of our
$H^\infty$ optimal control problem if we take a sufficiently large $N$.

The FSFH approximation of sampled-data $T_{zw}$
in FIg.~\ref{FSFHplant}
contains two sampling periods, $h$ and $h/N$,
and the whole system is periodically time-varying.
By using discrete-time lifting defined below,
the system can be equivalently converted to a
finite-dimensional discrete-time LTI system.
The discrete-time lifting is defined by
\[
 \begin{split}
 {\mathbf L}_N: 
&\{x[0],x[1],\ldots\}\\
&\quad \mapsto
\left\{ \begin{bmatrix}x[0]\\\vdots\\x[N-1]\end{bmatrix},
  \begin{bmatrix}x[N]\\\vdots\\x[2N-1]\end{bmatrix},\ldots
\right\},
\end{split}
\]
and its inverse by
\[
 \begin{split}
 {\mathbf L}_N^{-1}:
 &\left\{ \begin{bmatrix}x_1[0]\\\vdots\\x_N[0]\end{bmatrix},
  \begin{bmatrix}x_1[1]\\\vdots\\x_N[1]\end{bmatrix},\ldots
\right\}\\
 & \quad \mapsto \{x_1[0],\ldots,x_N[0],x_1[1],\ldots,x_N[1],\ldots\}.
 \end{split}
\]
By definition, discrete-time lifting ${\mathbf L}_N$ converts
a one-dimensional signal with sampling period $h/N$
to a $N$-dimensional signal with sampling period $h$.
Also, discrete-time lifting preserves the $\ell^2$ norm.
By using ${\mathbf L}_N$ and ${\mathbf L}_N^{-1}$,
we obtain a norm-equivalent discrete-time LTI system for 
time-varying ${\mathcal S}_{h/N}T_{zw}{\mathcal H}_{h/N}$.

Let $\mathbf{c2d}$ denote the step-invariant transformation \cite{CheFra},
that is,
\[
 \begin{split}
  {\mathbf{c2d}}\left(\ssr{A}{B}{C}{D},h\right)
    &:={\mathcal S}_{h}\ssr{A}{B}{C}{D}{\mathcal H}_{h}\\
  &=\ssr{e^{Ah}}{\int_0^h e^{At}Bdt}{C}{D},
 \end{split} 
\]
and $\mathbf{lift}$ denote the discrete-time lifting transformation \cite{CheFra},
that is,
\[
 \begin{split}
  &{\mathbf{lift}}\left(\ssr{A}{B}{C}{D},N\right)
  := {\mathbf L}_N\ssr{A}{B}{C}{D}{\mathbf L}_N^{-1}\\
  &\quad = \left[\begin{array}{c|cccc}
	 A^N & A^{N-1}B & A^{N-2}B & \ldots & B\\\hline
	 C & D & 0 & \ldots & 0\\
	 CA & CB & D & \ddots & \vdots\\
	 \vdots & \vdots & \vdots & \ddots & 0\\
	 CA^{N-1} & CA^{N-2}B& CA^{N-3}B & \ldots & D
    \end{array}\right].
 \end{split}
\]
Then we have the following theorem:
\begin{theorem}
	Assume that $L=mh+\frac{k}{N}h$ for some integers $m\geq 0$ and $k\in\{0,\ldots,N-1\}$. Then,
	for FSFH approximation ${\mathcal S}_{h/N}T_{zw}{\mathcal H}_{h/N}$,
	there exists a discrete-time LTI generalized plant $\Sigma_{dN}$ such that
	\[
	 \|{\mathcal S}_{h/N}T_{zw}{\mathcal H}_{h/N}\|_\infty = \|{\mathcal F}(\Sigma_{dN},K)\|_\infty,
	\]
	where the $H^\infty$ norm is defined by the $\ell^2$-induced norm.
	Moreover, norm-equivalent $\Sigma_{dN}$ is given by 
	\[
	  	\Sigma_{dN} := \begin{bmatrix}W_{dN}&-H_N\\ S_NW_{dN} & S_{N,k}z^{-m}P_{dN}G_{dN}H_N\end{bmatrix},
		\]
 where		
		\[
		\begin{split}
		W_{dN} &:= {\mathbf{lift}}\bigl({\mathbf{c2d}}(W,h/N),N\bigr),\\
		P_{dN} &:= {\mathbf{lift}}\bigl({\mathbf{c2d}}(P,h/N),N\bigr),\\
		G_{dN} &:= {\mathbf{lift}}\bigl({\mathbf{c2d}}(G,h/N),N\bigr),\\
		H_N &= [\underbrace{1,1,\ldots,1}_N]^\top,\quad
		S_N = [1,\underbrace{0,\ldots,0}_{N-1}],\\
		S_{N,k} &= [\underbrace{0 \ 0 \ \ldots \ 0}_{k} \ 1 \underbrace{\ 0 \ \ldots \ 0}_{N-k-1}].	
	  \end{split}
	\]

	\end{theorem}
	\begin{proof}
	 The formulae are obtained by the method described in \cite[Chap.~8]{CheFra}.
	\end{proof}
Note that $W_{dN}$, $P_{dN}$, and $G_{dN}$ are LTI.
Finally, our problem is reduced to a standard discrete-time $H^\infty$ control problem
with the LFT shown in Fig.~\ref{final plant}.
\begin{figure}[t]
	\centering
	\scalebox{0.8}{\includegraphics{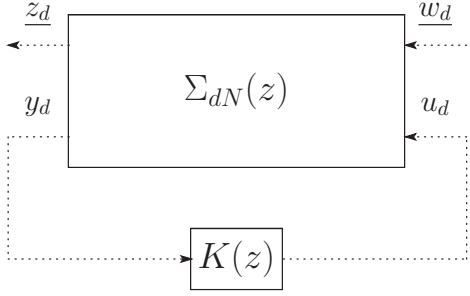}}
	\caption{Norm-equivalent discrete-time LTI system for FSFH approximation}
	\label{final plant}
\end{figure}
In this figure,
\[
  \underline{z_d}:= {\mathbf L}_Nz_d,\quad
  \underline{w_d}:= {\mathbf L}_Nw_d,\quad
  y_d := {\mathcal S}_h y,
\] 
and $u_d$ is the output of the controller $K(z)$.
Then the optimal controller $K(z)$ for this standard $H^\infty$ control problem
is easily obtained by using \texttt{hinfsyn} function in MATLAB Robust Control Toolbox.

%
%
\section{Simulation}
\label{sec:sim}
In this section, we show simulation results to illustrates the effectiveness of the proposed methods.

We assume that
sampling period $h$ is normalized to $1$.
The coupling wave path is modeled by
\[
	P(s) = \frac{0.25}{s+1},
\]
with time delay $L=1$, that is, the time delay is equal to the sampling period $h$.
The relay station is modeled by
$G(s) = 2.5$,
that is, the station amplifies input signals by 8 [dB].

For these parameters, we first design a feedforward canceler
proposed in Section \ref{sec:ff}.
We assume the frequency characteristic of input signals is given by
\[
 F(s) = \frac{1}{2s+1}.
\]
Note that the magnitude of $F(j\omega)$ represents the envelope of 
the spectra of the input signals (e.g. rectangular waves).
The discretization parameter for FSFH is set to $N=16$.
The obtained $H^\infty$-optimal $K(z)$ is of 18-th order.
With this filter, we simulate coupling wave canceling with a
periodic rectangular wave input with period $8h$.
Note that this signal contains frequency components beyond 
the Nyquist frequency, $\pi/h=\pi$ [rad/sec], although the frequency
of the periodic wave,
$\pi/8h=\pi/8$ [rad/sec] is much lower than $\pi$.

Fig.~\ref{fig:ff} shows the reconstructed signal $y$ (see Fig.~\ref{outside})
by the proposed feedforward canceler,
the input signal $v$, and the signal $y$ with no canceler.
\begin{figure}[t]
	\centering
	\includegraphics[width = 0.98\linewidth]{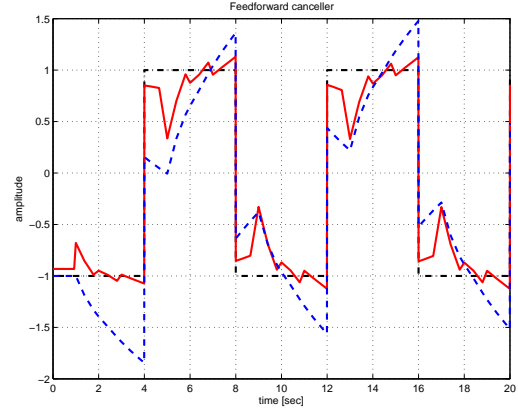}
	\caption{Feedforward cancelation: input signal (dash-dot line), 
	reconstructed signal $y$ by feedforward canceler (solid line), 
	signal $y$ with no canceler (dashed line)}
	\label{fig:ff}
\end{figure}
This result shows that the feedforward canceler works well.
To see this more precisely, we compute the effect of the coupling wave,
$|y(t)-v(t)|$, which is shown in Fig.~\ref{fig:ff_error}.
\begin{figure}[t]
	\centering
	\includegraphics[width = 0.98\linewidth]{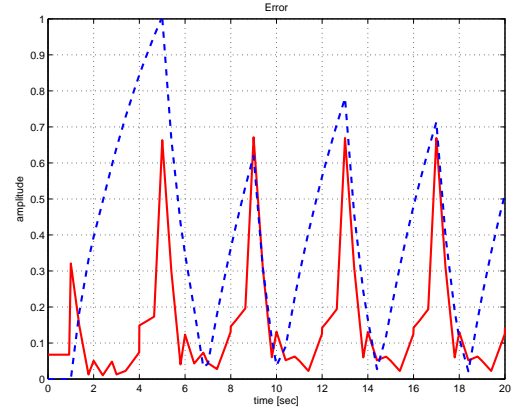}
	\caption{Coupling wave effect $|y(t)-v(t)|$
	by feedforward canceler (solid line) and
	with no canceler (dashed line)}
	\label{fig:ff_error}
\end{figure}
This result shows the proposed canceler well cancels the self-interference.

A drawback of the feedforward canceler is that it
never works if the gain of $G(s)$ is so high
that the feedback loop is unstable.
For example, if we take
$G(s) = 1000$,
that is, the relay station amplifies input signals by 60 [dB],
then the feedback loop becomes unstable.
For this situation, we adopt a feedback canceler proposed in Section \ref{sec:fb}.
The frequency characteristic $W(s)$ is assumed to be the same as $F(s)$,
that is, $W(s)=F(s)=1/(2s+1)$.
The other parameters are the same as those for the feedforward canceler design.
With FSFH discretization number $N=16$,
we compute the $H^\infty$-optimal $K(z)$ by the method
described in Section \ref{sec:FSFH}.

Fig.~\ref{fig:fb} shows the reconstructed signal $u$ in 
the feedback canceler (see Fig.~\ref{model0}).
\begin{figure}[t]
	\centering
	\includegraphics[width = 0.98\linewidth]{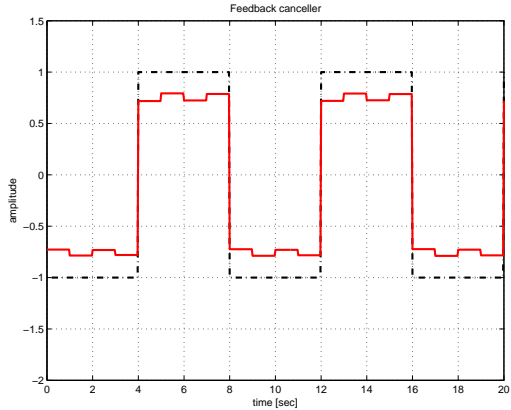}
	\caption{Feedback cancelation: input signal (dash-dot line) and 
	reconstructed signal $u$ by feedback canceler (solid line).}
	\label{fig:fb}
\end{figure}
Note that with any feedforward canceler,
the signal should diverge because the feedback loop
around the relay station itself is unstable.
On the other hand, the feedback canceler shows
small reconstruction error as shown in Fig.~\ref{fig:fb_error}.
\begin{figure}[t]
	\centering
	\includegraphics[width = 0.98\linewidth]{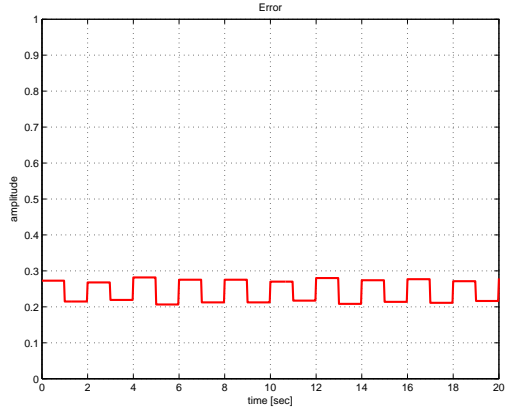}
	\caption{Coupling wave effect $|v(t)-u(t)|$
	by feedback canceler shown in Fig.~\ref{model0}.}
	\label{fig:fb_error}
\end{figure}

In summary, by the simulation results, sampled-data $H^\infty$ optimal design is 
proved to be effective
for coupling wave canceling.
%
%
\section{Conclusions}
\label{sec:conc}
In this article, we have proposed feedforward/feedback cancelers
based on the sampled-data $H^\infty$ control theory.
The design of feedforward cancelers is reduced to a 
sampled-data $H^\infty$ optimal discretization problem,
while that of feedback cancelers is formulated by
a standard sampled-data $H^\infty$ control problem.
They can be numerically solved by the FSFH method.
Simulation results have been shown to illustrates the effectiveness of the proposed 
feedforward/feedback cancelers.
Future work may include FIR (Finite Impulse Response) filter design as
proposed in \cite{NagYam14}, 
adaptive FIR filtering as proposed in
\cite{NagHamYam13}, 
robust filter design against uncertainty in the coupling wave path,
and implementation of the designed filter.

\hyphenation{KAKENHI}
\section*{Acknowledgment}
This work was supported in part by JSPS {KAKENHI} Grant Nos.\ 24360163, 24560457, 24560490,
24560543, and 26120521, and an Okawa Foundation Research Grant.


\end{document}